\documentclass[letterpaper, 10 pt, conference]{ieeeconf}  % Comment this line out if you need a4paper

\IEEEoverridecommandlockouts                              % This command is only needed if 
\usepackage{amsfonts}
\usepackage{amsmath}
\usepackage{subfig}
\usepackage{graphicx}
\usepackage{amssymb}
\usepackage{mathtools}
\usepackage{xcolor}
\usepackage{caption}
\usepackage{float}
\usepackage[ruled,vlined]{algorithm2e}
\usepackage{stmaryrd}

\newcommand{\col}{\textsf{ col}}
\newcommand{\norm}[1]{\left\lVert#1\right\rVert}

\newcommand{\R}{{\mathbb{R}}}

\newcommand{\Ce}{{\mathcal{C}}}
\newcommand{\N}{{\mathbb{N}}}

\newcommand{\e}{\mathsf{e}}

\newcommand{\X}{{\mathbf{X}}}
\newcommand{\T}{{\mathbf{T}}}
\newcommand{\Obs}{{\mathcal{U}}}
\newcommand{\U}{{\mathbf{U}}}
\newcommand{\sign}{{\mathsf{sign}}}

\newtheorem{theorem}{Theorem}[section]

\newtheorem{assumption}{Assumption}
\newtheorem{corollary}[theorem]{Corollary}

\newtheorem{lemma}[theorem]{Lemma}
\newtheorem{remark}[theorem]{Remark}
\newtheorem{problem}[theorem]{Problem}

\title{\LARGE \bf
Funnel-based Control for Reach-Avoid-Stay Specifications
}
\author{Ratnangshu Das and Pushpak Jagtap% <-this % stops a space
\thanks{*This work was not supported by any organization}% <-this % stops a space
\thanks{R. Das and P. Jagtap are with Robert Bosch Centre for Cyber-Physical Systems, IISc, Bangalore, India {\tt\small \{ratnangshud,pushpak\}@iisc.ac.in}}%
}

\begin{document}

\maketitle
\thispagestyle{empty}
\pagestyle{empty}

%%%%%%%%%%%%%%%%%%%%%%%%%%%%%%%%%%%%%%%%%%%%%%%%%%%%%%%%%%%%%%%%%%%%%%%%%%%%%%%%
\begin{abstract}

The paper addresses the problem of controller synthesis for control-affine nonlinear systems to meet reach-avoid-stay specifications. Specifically, the goal of the research is to obtain a closed-form control law ensuring that the trajectories of the nonlinear system, reach a target set while avoiding all unsafe regions and adhering to the state-space constraints. To tackle this problem, we leverage the concept of the funnel-based control approach. Given an arbitrary unsafe region, we introduce a circumvent function that guarantees the system trajectory to steer clear of that region. Subsequently, an adaptive funnel framework is proposed based on the target, followed by the construction of a closed-form controller using the established funnel function, enforcing the reach-avoid-stay specifications. To demonstrate the efficacy of the proposed funnel-based control approach, a series of simulation experiments have been carried out.

\end{abstract}

%%%%%%%%%%%%%%%%%%%%%%%%%%%%%%%%%%%%%%%%%%%%%%%%%%%%%%%%%%%%%%%%%%%%%%%%%%%%%%%%

%%%%%%%%%%%%%%%%%%%%%%%%%%%%%%%%%%%%%%%%%%%%%%%%%%%%%%%%%%%%%%%%%%%%%%%%%%%%%%%%
%%--------------------------------NEW SECTION---------------------------------%%
%%%%%%%%%%%%%%%%%%%%%%%%%%%%%%%%%%%%%%%%%%%%%%%%%%%%%%%%%%%%%%%%%%%%%%%%%%%%%%%%
\section{Introduction}

In recent years, there has been significant interest in the study of reach-avoid-stay (RAS) specifications for the safe and reliable operation of autonomous systems. Essentially the system state trajectory should eventually reach a target set while avoiding any unsafe set and respecting state space constraints. Synthesizing controllers for these RAS specifications is an important class of control problem as they serve as building blocks for complex task specifications 
\cite{Kloetzer} and enable the design of robust control strategies in safety-critical control problems such as trajectory regulation, motion planning, and obstacle avoidance. 

With the onset of the usage of formal languages for specifying complex tasks, symbolic control \cite{tabuda, FRR} has emerged as a powerful tool. 
%SCOTS \cite{SCOTS} has gained significant attention in the research community for its ability to leverage abstraction techniques to create a symbolic model of the system dynamics and efficiently compute control strategies with formal guarantees. 
\cite{FPA_RAS} proposed a fixed-point algorithm as a computational improvement over the abstraction-based methods for control synthesis in a reach-stay scenario. The authors in \cite{Local_Global} presented a scalable controller synthesis technique by leveraging the concept of barrier functions in symbolic control. In spite of all these attempts at improving computational efficiency, these approaches still face challenges related to the so-called curse of dimensionality.

In contrast to formal methods, nonlinear control approaches like barrier-based control \cite{CBF} ensure formal guarantees of safety and stability without the need for state-space discretization. Authors in \cite{Meng1} proposed implementing control Lyapunov-barrier functions to establish sufficient conditions for reach-avoid-stay specifications, specifically in the context of a system experiencing a Hopf-bifurcation. In \cite{Meng2}, researchers present a stochastic analog of Lyapunov-barrier functions to characterize probabilistic reach-avoid-stay specifications, taking robustness into account. However, although these methods provide more efficient control synthesis, the reliance on optimization techniques can still lead to increased computational complexity, making barrier function-based methods computationally demanding, especially for large and high-dimensional systems.

On the other hand, the funnel-based control approach \cite{PPC1} offers the distinct advantage of designing a closed-loop control scheme satisfying a required tracking performance. Owing to the computationally tractable nature of funnel-based control, numerous successful applications have been reported in the literature \cite{surveyPPC}. From solving tracking control problems for unknown nonlinear systems \cite{PPC_Unknown} to handling multi-agent systems subjected to complex task specifications \cite{Funnel_STL_MAS}, researchers have demonstrated its efficacy in a wide range of control problems. Moreover, as the feedback control algorithm actively adjusts the system's trajectory to guide it towards the target, it has been effective in enforcing reachability specifications, i.e., reaching a target while respecting state space constraints \cite{hard_soft, NAHS}.

However, active obstacle avoidance using funnel-based control can be a challenging problem. One of the main difficulties lies in designing accurate and efficient funnel representations to ensure safe navigation around obstacles while maintaining reach-avoid-stay specifications. In \cite{RB}, authors consider a pre-established trajectory around the obstacles and redefine the problem as implementing control funnel functions for path following. 
Although this approach ensures that the system remains in a safe region around the reference trajectory, it fails to utilize the inherent ability of funnel constraints to avoid obstacles.

This paper puts forward, for the very first time, a novel approach to integrate avoid-specifications within the funnel-based control framework. By adapting the funnel constraints, the closed-form control law dynamically adjusts the robot's trajectory to avoid any general unsafe set while maintaining the desired performance criteria. The effectiveness of this approach in satisfying reach-avoid specifications is further demonstrated through simulation studies, highlighting its potential to enhance the capabilities of robotic systems in navigating complex environments. 

%The remainder of this paper is structured as follows. In Section \ref{sec:prob}, we introduce the system dynamics and formulate the problem statement considered in this work. Section \ref{sec:reach} elaborates on implementing a funnel-based control approach to solve reachability specifications. In Section \ref{sec:avoid}, we introduce the circumvent function as a way to incorporate avoid specifications within the funnel constraint design. Section \ref{sec:control} presents the design of a closed-form control policy by utilizing the results obtained in Section \ref{sec:avoid} to enforce reach-avoid specifications. In Section \ref{sec:multiobs}, we put forward an algorithm to address avoiding general concave and disconnected unsafe sets using the control law established in \ref{sec:control}. Finally, in Section \ref{sec:sim}, the effectiveness of this approach is demonstrated through simulation studies, highlighting its potential to enhance the capabilities of robotic systems in navigating complex environments. 

%%%%%%%%%%%%%%%%%%%%%%%%%%%%%%%%%%%%%%%%%%%%%%%%%%%%%%%%%%%%%%%%%%%%%%%%%%%%%%%%
%%--------------------------------NEW SECTION---------------------------------%%
%%%%%%%%%%%%%%%%%%%%%%%%%%%%%%%%%%%%%%%%%%%%%%%%%%%%%%%%%%%%%%%%%%%%%%%%%%%%%%%%
\section{Preliminaries and Problem Formulation} \label{sec:prob}

\subsection{Notations}
The symbols $\N$, $ \R$, $\R^+$, and $\R_0^+ $ denote the set of natural, real, positive real, and nonnegative real numbers, respectively. 
We use $ \R^{n\times m} $ to denote a vector space of real matrices with $ n $ rows and $ m $ columns. To represent a column vector with $n$ rows, we use $ \R^{n}$.
We represent the Euclidean norm using $\|\cdot\|$. For $a,b\in\R$ and $a< b$, we use $(a,b)$ to represent open interval in $\R$. For $a,b\in\N$ and $a\leq b$, we use $[a;b]$ to denote close interval in $\N$. To denote a vector $x \in \R^{n}$ with entries $x_1, \ldots, x_n$, we use $\col(x_1, \ldots, x_n)$, where $x_i \in \R, i \in [1;n]$ denotes $i$-th element of vector $x\in\R^n$. 
%We use $I_n$ and $0_{n\times m}$ to denote identity matrix in $\R^{n\times n}$ and zero matrix in $\R^{n\times m}$, respectively. 
A diagonal matrix in $\R^{n\times n}$ with diagonal entries $d_1,\ldots, d_n$ is denoted by $\textsf{diag}(d_1,\ldots, d_n)$.
% Given a matrix $M\in\R^{n\times m}$, $M^T$ represents transpose of matrix $M$. Given a matrix $P\in\R^{n\times n}$, $\Tr(P)$ represents trace of matrix $P$. 
% Given a set $A$, we use $|A|$ to represent the cardinality of the set $A$. 
Given $N \in \N$ sets $\X_i$, $i\in\left[1;N\right]$, the Cartesian product of the sets is given by $\X=\prod_{i\in\left[1;N\right]}\X_i:=\{(x_1,\ldots,x_N)|x_i\in \X_i,i\in\left[1;N\right]\}$.
Consider a set $\X_a\subset\R^n$, its projection on $i$th dimension, where $i\in[1;n]$, is given by an interval $[\underline \X_{ai},\overline \X_{ai}]\subset \R$, where $\underline \X_{ai}:=\min\{x_i\in\R\mid[x_1\ldots,x_n]\in \X_a\}$ and $\overline \X_{ai}:=\max\{x_i\in\R\mid[x_1,\ldots,x_n]\in \X_a\}$. We further define the hyper-rectangle $\llbracket \X_a \rrbracket = \prod_{i=[1;n]}{[\underline \X_{ai}, \overline \X_{ai}]}$.
We denote the empty set by $\emptyset$. The space of bounded continuous functions is denoted by $\Ce$.
Given a compact set $\X$, $int(\X)$ represents the interior of the set and 
$\partial \X = \X \setminus int(\X)$ represents the boundary of $\X$.
$\overline{\max}$ and $\overline{\min}$ are smooth approximations of the non-smooth $\max$ and $\min$ functions, defined as, $\overline{\max}(a,b) \approx \frac{1}{\nu}\ln(\e^{\nu a}+\e^{\nu b})$ and $\overline{\min}(a,b) \approx -\frac{1}{\nu}\ln(\e^{-\nu a}+\e^{-\nu b})$, respectively. The sign function is defined as
$
\sign(x) :=
\begin{cases}
  -1 & \text{if } x < 0 \\
  1  & \text{if } x \geq 0
\end{cases}
$.

\subsection{System Definition}
Consider the following control-affine nonlinear system:
\begin{align}
    \mathcal{S}: \dot{x} = f(x) + g(x)u, \label{eqn:sysdyn}
\end{align}
where $x(t) = \col(x_1(t), \ldots, x_n(t)) \in \X \subset \mathbb{R}^n$ and $u(t) \in \mathbb{R}^m$ are the state and control input vectors, respectively. The state space of the system is defined by the closed and connected set $\X$. 
%$w(t) \in \mathbb{W} \subset \mathbb{R}^n$ is the additive noise of a nonlinear system, where $\mathbb{W}$ is a compact set. 
The functions $f: \X \rightarrow \mathbb{R}^n$ and $g: \X \rightarrow \mathbb{R}^{n \times m}$ satisfy Assumption 1.
\begin{assumption} \label{assum:lip}
$f$ and $g$ are locally Lipschitz, and $g(x)g^T(x)$ is positive definite for all $x \in \mathbb{R}^n$.
\end{assumption}

\subsection{Problem Formulation}
The paper considers the desired behavior of the system $\mathcal{S}$, in (\ref{eqn:sysdyn}), defined in the form of reach-avoid-stay specifications. 

Let the compact and connected set $\T \subset \X$ be the target set, the set $\U \subset \X$ be an unsafe region containing $n_u \in \N$ unsafe sets
%, satisfying Assumption \ref{assum:obs}, 
defined as, $\U = \bigcup_{j \in [1;n_u]} \Obs^j$, where $\Obs^j \subset \X$ is a convex, compact, and connected set, representing, the $j$th unsafe set. 
Thus, in general, the unsafe region $\U$, although necessarily compact, can be disconnected and nonconvex.
%The reach-avoid-stay specification is defined as: the trajectories of the system starting in $\X \setminus \U$ should reach target $\T$ while avoiding unsafe region $\U$ and remaining in $\X$. 

Now, we will formally define the main controller synthesis problem considered in this work.

\begin{problem}\label{prob1}
    Given a control-affine system $\mathcal{S}$ in (\ref{eqn:sysdyn}) with Assumption \ref{assum:lip}, target set $\T \subset \X$, and unsafe region $\U$, as defined above, design a closed-form controller 
    %$u:\X \rightarrow \R^m$ 
    to ensure the satisfaction of the reach-avoid-stay specification, i.e.,
    for a given initial position $x(0) \in \X \setminus \U$, there exists $t \in \R_0^+$, such that, $x(t) \in \T$ and for all $t \in \R_0^+ : x(t) \in \X \setminus \U$.
\end{problem}
%Thus, the specification ensures that the system will eventually reach the target set $\T$, while staying in $\X$ and never entering the unsafe region $\U$.

We approach this problem using a funnel-based control strategy to enforce reachability specification (Section \ref{sec:reach}) and then dynamically modifying the funnel around the unsafe region to ensure that the system trajectory avoids the unsafe region while respecting the state constraints (Sections \ref{sec:avoid}-\ref{sec:multiobs}).

\begin{remark}
    If $\X$ is of any arbitrary shape, we redefine the state space as the hyper-rectangle $\hat{\X} := \llbracket \X \rrbracket = \prod_{i \in [1;n]} [\underline{\X}_i, \overline{\X}_i]$ and expand the unsafe region $\hat{\U} = \U \cup \left(\llbracket \X \rrbracket \setminus \X\right)$.
    Here, $[\underline{\X}_i, \overline{\X}_i]$ represent the projection of set $\X$ on the $i$th dimension.
    Note that, adding $\llbracket \X \rrbracket \setminus \X$ to the unsafe set $\U$ and following Algorithm 1 in Section \ref{sec:multiobs}, enforces stay specifications for an arbitrary state-space $\X$.
    %Hence, without loss of generality, we will assume $\X = \bigcup_{i \in [1;n]} [\underline{\X}_i, \overline{\X}_i]$ in rest of the paper.
\end{remark}
%%%%%%%%%%%%%%%%%%%%%%%%%%%%%%%%%%%%%%%%%%%%%%%%%%%%%%%%%%%%%%%%%%%%%%%%%%%%%%%
%%%%%%%%%%%%%%%%%%%%%%%%%%%%%%%%%%%%%%%%%%%%%%%%%%%%%%%%%%%%%%%%%%%%%%%%%%%%%%%

\section{Controller for Reachability Specification}\label{sec:reach}
In this section, we formulate a funnel-based control strategy aimed at guaranteeing that the system's trajectory adheres to the reachability specifications, i.e., given a target set $\T \subset \X$ and a given initial position $x(0) \in \X$, the controlled trajectory will eventually reach the target set in finite time. To solve the reachability problem, we leverage the funnel-based control approach \cite{PPC1}. We first define the funnel constraints over the trajectory as follows:
%\subsection{Funnel Control}
%Funnel-based control \cite{PPC1} is a powerful control technique that employs dynamically varying regions, referred to as funnels, to guide a system's trajectory. In accordance to our reachability problem, we define the funnel constraints over the trajectory as follows:
\begin{gather}
    \underbrace{-\underline{c}_i \rho_i(t) + \eta_i}_{\rho_{i,L}(t)} < x_i(t) < \underbrace{\overline{c}_i \rho_i(t) + \eta_i}_{\rho_{i,U}(t)}, \forall i \in [1;n], \label{eqn:ppc}
\end{gather}
where $\eta = \col(\eta_1, \ldots, \eta_n) \in int(\T)$, $\underline{c}_i = \eta_i - \underline{\X}_{i}$ and $\overline{c}_i = \overline{\X}_{i} - \eta_i$. 
%ensures that the funnel always guides the system trajectory within the state space limits: $[\underline{\X}_{i}, \overline{\X}_{i}], \forall i \in [1;n]$.
$\rho_i(t)$ is the continuously differentiable, positive and non-increasing funnel function defined as:
\begin{align}
    \rho_i(t) = (\rho_{i,0} - \rho_{i,\infty})e^{-l_i t} + \rho_{i,\infty} \label{eqn:perform}
\end{align}
with $\rho_{i,0}=1$, $\rho_{i,\infty} \in \left(0, \min \left( \rho_{i,0}, \frac{|\T_i-\eta_i|}{\max\{\underline{c}_i, \overline{c}_i\}}\right)\right)$ and $l_i \in \R_0^+$ governs the lower bound of convergence rate. 

The above choice of $\rho_{i,0}$, $\underline{c}_i$, and $\overline{c}_i$ ensures that the initial state of the system $x_i(0)$ is within $[\underline{\X}_{i}, \overline{\X}_{i}], \forall i \in [1;n]$ and by the aforementioned choice of $\rho_{i,\infty}$, as $t \rightarrow \infty, x(t) \in \prod_{i \in [1;n]} \left( \eta_i + [-\underline{c}_i\rho_{i,\infty}, \overline{c}_i\rho_{i,\infty}] \right) \subset \T$. Thus, enforcing system state inside funnel constraints \eqref{eqn:ppc} ensures reachability.
% prescribes a constrain on the system trajectory, leading to a neighborhood around the point $\eta$, $nbd(\eta) = \{x \in \R^n : -\underline{c}_i\rho_{i,\infty} < x_i-\eta_i < \overline{c}_i\rho_{i,\infty}\, \forall i \in [1;n]\}$, such that $nbd(\eta) \in \T$.
%Thus, we ensure that the initial state of the system $x_i(0)$ is within the funnel, i.e., $-\underline{c}_i\rho_{i,0} < x_i(0)-\eta_i < \overline{c}_i\rho_{i,0}$ and the system trajectory eventually reaches the target $\T$.
An example of a funnel designed for enforcing reachability specification is shown in Figure \ref{fig:Funnel} (a).

%%%%%%%%%%%%%%%%%%%%%%%%%%%%%%%%%%%%%%%%%%%%%%%%%%%%%%%%%%%%%%%%%%%%%%%%%%%%%%%
%%%%%%%%%%%%%%%%%%%%%%%%%%%%%%%%%%%%%%%%%%%%%%%%%%%%%%%%%%%%%%%%%%%%%%%%%%%%%%%

%\subsection{Controller Design and Stability Analysis}
% The controller design for a nonlinear control-affine system to fulfill reach-stay tasks is done in two stages.

% In stage I, we design the funnel. Given an initial state $x(0)$ and target state $\T$, we choose initial funnel parameters $\rho_{i,0}, \rho_{i,\infty} \text{ and } l$ and construct $\rho_i(t)$ as given in Equation \ref{eqn:perform}.

% Next, we proceed with establishing the control law. First, we normalize the difference system state $x(t)$ and target state $\eta$ with respect to the funnel to obtain the normalized error $\hat{e}(x,t) = \col(e_i(x,t))$, where
To design a controller enforcing condition \eqref{eqn:ppc}, we first define the normalized error ${e}(x,t) = \col(e_1(x_1,t),\ldots e_n(x_n,t))$, as
% \begin{align}
%     e_i(x,t) = \frac{x_i(t)-\eta_i}{\rho_{i}(t)}, \forall i \in [1;n]. \label{eqn:e_rs}
% \end{align}
\begin{align}
    e_i(x_i,t) = \frac{x_i(t)-\frac{1}{2}(\rho_{i,U}(t) + \rho_{i,L}(t))}{\frac{1}{2}(\rho_{i,U}(t) - \rho_{i,L}(t))}, \forall i \in [1;n]. \label{eqn:e_rs}
\end{align}
Now the corresponding constrained region $\mathbb{D}$ can be represented by $\mathbb{D} := \{{e}(x,t) : e_i(x_i,t) \in (-1, 1), \forall i \in [1;n]\}$. 
Next the normalized error is transformed through a  smooth and strictly
increasing transformation function $y: \mathbb{D} \rightarrow \mathbb{R}^n$ with $y(0) = 0$. The transformed error is then defined as $\varepsilon = \col(\varepsilon_1, \ldots, \varepsilon_n)$, where
\begin{align}
    \varepsilon_i(x,t) = y(e_i(x,t)) = \ln\left(\frac{1+e_i(x,t)}{1-e_i(x,t)}\right), \forall i \in [1;n]. \label{eqn:eps_rs}
\end{align}
By this definition, if the transformed error $\varepsilon(x,t)$ is bounded, then the normalized error ${e}(x,t)$ is confined within the constrained region $\mathbb{D}$ and the state $x(t)$ adheres to (\ref{eqn:ppc}).
We also define $\xi(x,t) = \textsf{diag}(\xi_1(x,t), \ldots, \xi_n(x,t))$ with 
\begin{align}
    %\xi_i(x,t) = \frac{2}{\rho_{i}(t)(1-e_i(x,t)^2)}, \forall i \in [1;n]. 
    \xi_i(x,t) = \frac{4}{\rho_{i,d}(t)(1-e_i(x,t)^2)}, \forall i \in [1;n]
    \label{eqn:xi_rs}
\end{align}
where $\rho_{i,d} = \rho_{i,U} - \rho_{i,L}$. 

Now, in Theorem \ref{thm:control}, we propose a control strategy $u(x,t)$ such that the state trajectory is constrained within the funnel.

\begin{theorem} \label{thm:control}
    Consider the control-affine system $\mathcal{S}$ given in (\ref{eqn:sysdyn}) with Assumptions \ref{assum:lip}. 
    Given a target set $\T$, the funnel constraints $\rho_{i,U}(t)$ and $\rho_{i,L}(t)$ (\ref{eqn:ppc}), the control strategy
    \begin{multline} \label{eqn:Control_rs}
        u(x,t) = -g(x)^T(g(x) g(x)^T)^{-1} \\
        \left(k\xi(x,t) \varepsilon(x,t) - \frac{1}{2}\dot{\rho}_d(t) {e}(x,t) \right)
    \end{multline}
    will drive the state trajectory $x(t)$, to the target set $\T$ in finite time, i.e., $\exists t \in \R_0^+ : x(t) \in \T$.
    Here, $k$ is any positive constant,  $\rho_d := \textsf{diag}(\rho_{1,d},\ldots,\rho_{n,d})$, with $\rho_{i,d} = \rho_{i,U} - \rho_{i,L}$, ${e}(x,t)$, $\varepsilon(x,t)$, and $\xi(x,t)$ are defined in \eqref{eqn:e_rs}, \eqref{eqn:eps_rs}, and \eqref{eqn:xi_rs}, respectively.
\end{theorem}

\begin{proof}
    The proof follows on similar grounds as that of Theorem \ref{theorem_ras} and is omitted here due to space constraints.
\end{proof}

Thus, given a system $\mathcal{S}$ in (\ref{eqn:sysdyn}), a target set $\T$ in the state space $\X$, we can define a funnel and the closed-form well-defined control law (\ref{eqn:Control_rs}) that will guide the system trajectory to the target, enforcing reachability specifications.

%%%%%%%%%%%%%%%%%%%%%%%%%%%%%%%%%%%%%%%%%%%%%%%%%%%%%%%%%%%%%%%%%%%%%%%%%%%%%%%%
%%--------------------------------NEW SECTION---------------------------------%%
%%%%%%%%%%%%%%%%%%%%%%%%%%%%%%%%%%%%%%%%%%%%%%%%%%%%%%%%%%%%%%%%%%%%%%%%%%%%%%%%
\section{Extension to Reach-Avoid-Stay Specification}\label{sec:avoid}
In this section, we begin by exploring the integration of avoidance of unsafe regions within the funnel-based control framework. Subsequently, we present an adaptive funnel design strategy that enables the successful accomplishment of reach-avoid-stay tasks.

\begin{figure*}[h]
    \centering
    \includegraphics[width=\textwidth]{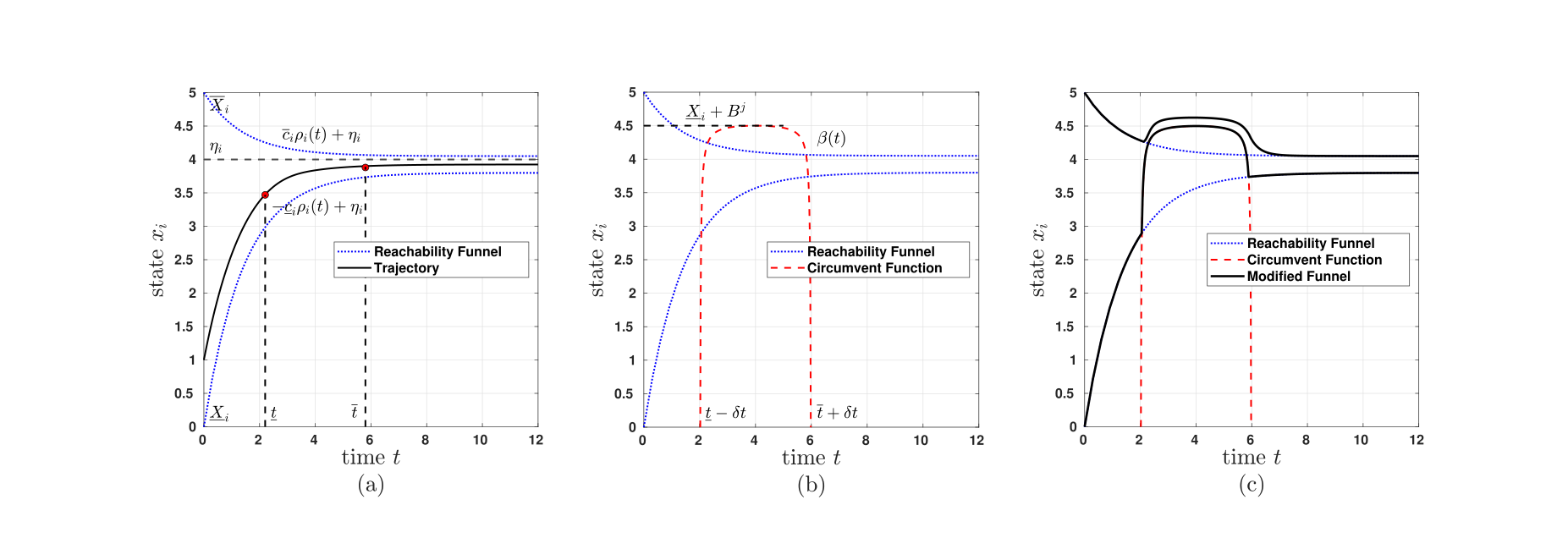}
    \vspace{-0.8cm}
    \caption{Funnel Design. (a) Reachability funnel to obtain $\underline{t}$ and $\overline{t}$. (b) Introduction of circumvent function. (c) Funnel adapted around circumvent function.}
    \label{fig:Funnel}
\end{figure*}

\subsection{Design of Circumvent Function}\label{sec:bump}
Consider an unsafe region $\U$ with $n_u$ compact, connected and convex sets $\Obs^j$, for $j \in [1;n_u]$. We propose to introduce the avoid specifications through a circumvent function $\beta^j(t)$, $j \in [1;n_u]$. 
\begin{remark}
    Note that although we are putting an assumption on $\Obs^j$ to be convex and connected, the general unsafe zone $\U$ can be concave and disconnected. This will further be elaborated upon in Section \ref{sec:multiobs}.
\end{remark}
%$\beta^j_i$ corresponds to $i$th dimension and $j$th obstacle, with $i \in [1;n]$ and $j \in [1;n_u]$.

First, given an initial state $x(0) \in \X \setminus \U$, we obtain the time range $[\underline{t}^j, \overline{t}^j]$ over which the system trajectory $x(t)$, on application of the control law $u(x,t)$ (\ref{eqn:Control_rs}) to satisfy the reachability specification, intersects with the $j$th unsafe set $\Obs^j$, and is given by,
$
    \underline{t}^j = \inf \{ t \in \R^+ : x(t) \cap \Obs^j \neq \emptyset \}$ and $\overline{t}^j = \sup \{ t \in \R^+ : x(t) \cap \Obs^j \neq \emptyset \}$.

Consider the first unsafe set that the system trajectory intersects be $\Obs^{\hat{j}}$, where
\begin{align} \label{eqn:firstobs}
    \hat{j} = \arg \min_{j \in [1;n_u]} \underline{t}^j_i.
\end{align}
Following this, we will discuss the introduction of the circumvent function and adaptive funnel design to steer clear of $\Obs^{\hat{j}}$. The subsequent extension to deal with the entire unsafe region $\U$ with multiple disconnected concave unsafe sets is presented in Section \ref{sec:multiobs}.

Further, note that the system's trajectory enters the unsafe zone $\Obs^{\hat{j}}$, if and only if $\exists t \in \R^+$, such that $x_i(t) \cap [\underline{\Obs}^{\hat{j}}_i, \overline{\Obs}^{\hat{j}}_i] \neq \emptyset, \forall i \in [1;n]$.
Hence, to satisfy the avoid specification, it is sufficient to introduce the circumvent function only in one dimension $i^{\hat{j}}$, given by
\begin{gather}\label{eqn:firstdim}
    i^{\hat{j}} = \arg \min_{i \in [1;n]}{t_i^{\hat{j}}}, 
\end{gather}
% \begin{gather}
%     i^{\hat{j}} = \arg \min_{i} \left( \inf \{ t \in \R^+ : x(t) \cap [\underline{\Obs}^{\hat{j}}_i, \overline{\Obs}^{\hat{j}}_i] \neq \emptyset \} \right), \forall i = [1;n]
% \end{gather}
where $t_i^{\hat{j}} = \inf \{ t \in \R^+ : x(t) \cap [\underline{\Obs}^{\hat{j}}_i, \overline{\Obs}^{\hat{j}}_i] \neq \emptyset \}$ and $[\underline{\Obs}^{\hat{j}}_i, \overline{\Obs}^{\hat{j}}_i]$ is the projection of $\Obs^{\hat{j}}$ in the $i$th dimension. Note that, $i^{\hat{j}}$ may not be unique and in the case, the trajectory enters the projections of $\Obs^{\hat{j}}$ in multiple dimensions at the same time, the $\arg \min$ function returns $i^{\hat{j}}$ randomly from those multiple alternatives. The advantage of this random selection will be discussed in Section \ref{sec:multiobs}.

We also have the liberty to choose between modifying either the upper or the lower constraint boundary of the funnel. Unless the scenario where $\underline{\Obs}^{\hat{j}}_i = \underline{\X}_{i}$ or $\overline{\Obs}^{\hat{j}}_i = \overline{\X}_{i}$, where the circumvent function should necessarily be introduced in the lower and upper constraint boundary, respectively (it can be visualized as the scenario of a wall-shaped obstacle, where there is no space between the state space boundary and the obstacle at one end), we randomly choose between the two options. Although an optimal alternative can be easily chosen, the advantage of random picking is discussed in Section \ref{sec:multiobs}.

%We, finally, decide whether to introduce the circumvent function on the lower constraint boundary or the upper constraint boundary of the funnel. For $i = i^{\hat{j}}$, the circumvent is necessarily added to the lower constraint if, $\underline{\Obs}^{\hat{j}}_i = \underline{\X}_{i}$, and to the upper constraint if, $\overline{\Obs}^{\hat{j}}_i = \overline{\X}_{i}$. This can be visualised as the scenario of a wall-shaped obstacle, where there is no space between the state space boundary and the obstacle at one end.
%Otherwise, randomly choose between the two options. Although an optimal alternative can be easily chosen, the advantage of random picking is discussed in Section \ref{sec:multiobs}.

% For $i = i^{\hat{j}}$, the circumvent is added to the lower constraint if, (1) $\frac{\overline{\Obs}^{\hat{j}}_i + \underline{\Obs}^{\hat{j}}_i}{2} < \eta_i$ or (2) $\underline{\Obs}^{\hat{j}}_i = \underline{\X}_{ai}$. Similarly, for $i = i^{\hat{j}}$, the circumvent is added to the upper constraint if, (1) $\frac{\overline{\Obs}^{\hat{j}}_i + \underline{\Obs}^{\hat{j}}_i}{2} > \eta_i$ or (2) $\overline{\Obs}^{\hat{j}}_i = \overline{\X}_{ai}$. The second cases consider the scenario of a wall, when there is no space between the state space boundary and the obstacle at one end.

We define the circumvent function $\beta$ on lower constraint boundary for $i = i^{\hat{j}}$ as: 
\begin{equation} \label{eqn:bump1}
\beta^{\hat{j}}_i(t) =
    \begin{cases}
        B^{\hat{j}} \e^{\frac{-k^{\hat{j}} \left( t- m^{\hat{j}} \right)^2}{(r^{\hat{j}})^2 - \left(t-m^{\hat{j}}  \right)^2}} + \underline{\X}_{i}, & \forall t \in T_{act}\\
        \underline{\X}_{i}, & \forall t \in \R^+ \setminus T_{act}
    \end{cases}
\end{equation}
where, $B^{\hat{j}} = \overline{\Obs}^{\hat{j}}_i - \underline{\X}_{ai} + \delta B$, $m^{\hat{j}} := \frac{\underline{t}^{\hat{j}} + \overline{t}^{\hat{j}}}{2}$, $r^{\hat{j}} := \frac{\underline{t}^{\hat{j}} - \overline{t}^{\hat{j}}}{2} + \delta t$ and $\delta t \in \R^+$ is a tolerance factor. The function is active in the time range $T_{act} = [\underline{t}^{\hat{j}}-\delta t, \overline{t}^{\hat{j}}+\delta t]$ when the system trajectory avoids $\Obs^{\hat{j}}$. The 
$\delta B$ governs how far from $\Obs^{\hat{j}}$ should the trajectory stay clear.
$k^{\hat{j}} \in \R^+$ is a small positive constant and determines the smoothness of the circumvent function. 

Similarly, we define a circumvent on the upper constraint boundary as
\begin{equation} \label{eqn:bump2}
\beta^{\hat{j}}_i(t) =
    \begin{cases}
        -B^{\hat{j}} \e^{\frac{-k^{\hat{j}} \left( t- m^{\hat{j}} \right)^2}{(r^{\hat{j}})^2 - \left(t-m^{\hat{j}}  \right)^2}} + \overline{\X}_{ai}, & \forall t \in T_{act}\\
        \overline{\X}_{ai}, & \forall t \in \R^+ \setminus T_{act}
    \end{cases}
\end{equation}
with $B^{\hat{j}} = \overline{\X}_{ai} - \underline{\Obs}^{\hat{j}}_i + \delta B$
and the rest of the parameters are the same as above.

An example of the introduction of the circumvent function on the lower constraint of a funnel is shown in Figure \ref{fig:Funnel} (b).

\subsection{Adaptive Funnel Design} \label{Funnel}
Given a target set $\T$ and obstacle $\Obs^{\hat{j}}$, choose a point $\eta \in int(\T\setminus U)$. Now, according to \eqref{eqn:ppc}, construct the funnel $\rho_L$ and $\rho_U$ to satisfy the reachability specification. As defined in the previous subsection, we characterize the obstacle using the circumvent function $\beta^{\hat{j}} (t)$, and now, we incorporate it into the funnel design. To solve Problem \ref{prob1}, we propose the following adaptive funnel constraints.
%Consider a reach and avoid task where the reach specifications and obstacle avoidance are captured by $\rho(t)$ (3) and $\beta(t)$ (4) respectively. After incorporating $\beta(t)$ into $\rho(t)$, the revamped funnel is defined as:
\begin{align}
   \hspace{-0.5em} \begin{matrix}
\text{If } \beta^{\hat{j}}_i\text{ introduced on}\\ \text{lower constraint } \rho_{i,L}
\end{matrix}\hspace{-0.3em}:\hspace{-0.3em}
\left\{\begin{matrix}
\gamma_{i,L} (t) := \overline{\max}(\rho_{i,L}(t), \beta^{\hat{j}}_i(t)),\\
    \gamma_{i,U} (t) := \rho_{i,U}(t) + \alpha_i(t),\hfill
\end{matrix}\right.\\
%\end{align}
%\begin{align}
    \hspace{-0.55em}\begin{matrix}
\text{If } \beta^{\hat{j}}_i\text{ introduced on}\\ \text{upper constraint } \rho_{i,U}\\ 
\end{matrix}\hspace{-0.3em}:\hspace{-0.3em}
\left\{\begin{matrix}
\gamma_{i,L} (t) := \rho_{i,L}(t) - \alpha_i(t),\hfill\\
    \gamma_{i,U} (t) := \overline{\min}(\rho_{i,U}(t), \beta^{\hat{j}}_i(t)),
\end{matrix}\right.
\end{align}
% \begin{subequations}
% \begin{align}
%     \gamma_{i,L} (t) &:= \overline{\max}(\rho_{i,L}(t), \beta^{\hat{j}}_i(t)),\\
%     \gamma_{i,U} (t) &:= \rho_{i,U}(t) + \alpha_i(t),
% \end{align}
% \end{subequations}
% if the circumvent is introduced on the lower constraint boundary. Or
% \begin{subequations}
% \begin{align}
%     \gamma_{i,L} (t) &:= \rho_{i,L}(t) - \alpha_i(t),\\
%     \gamma_{i,U} (t) &:= \overline{\min}(\rho_{i,U}(t), \beta^{\hat{j}}_i(t)),
% \end{align}
% \end{subequations}
% if the circumvent is introduced on the upper constraint boundary.
The modifications in the constraints of the funnel are captured by a continuously differentiable update function, $\alpha(t) = \col(\alpha_1(t), \ldots, \alpha_n(t))$. The adaptive law governing the dynamics of the update function is defined as:
\begin{align}\label{eqn:adap}
    \dot{\alpha}_i(t) = \frac{\theta_i(t)}{\psi_i(t) + \alpha_i(t)} - \kappa\alpha_i(t), \ \alpha_i(0) = 0,
\end{align}
where $\psi_i(t) = \rho_{i,U}(t)-\beta^{\hat{j}}_i(t)-\mu$ if  $\beta^{\hat{j}}_i$ introduced on lower constraint $\rho_{i,L}$ and $\psi_i(t) = \beta^{\hat{j}}_i(t)-\rho_{i,L}(t)-\mu$ if  $\beta^{\hat{j}}_i$ introduced on upper constraint $\rho_{i,U}$, with $\mu \in \R^+$ as a tolerance factor. 
% \begin{align}
%     \psi_i(t):=\left\{\begin{matrix}
%  \rho_{i,U}(t)-\beta^{\hat{j}}_i(t)-\mu,& \text{If } \beta^{\hat{j}}_i\text{ introduced on lower constraint } \rho_{i,L} \\ 
% \beta^{\hat{j}}_i(t)-\rho_{i,L}(t)-\mu, & \text{If } \beta^{\hat{j}}_i\text{ introduced on lower constraint } \rho_{i,L}
% \end{matrix}\right.
% \end{align}
% if the circumvent is introduced on the lower constraint boundary, or
% \begin{align}
%     \psi_i(t) = \beta^{\hat{j}}_i(t)-\rho_{i,L}(t)-\mu \nonumber
% \end{align}
% if the circumvent is introduced on the upper constraint boundary, 
$\theta_i(t)$ acts as a trigger, activating the first part of the update function only when reach-avoid specifications are conflicting with a tolerance of $\mu$ and is given by:
\begin{align}
    \theta_i(t) = \theta_o(1-\sign(\psi_i(t))), \nonumber
\end{align}
where $\theta_o \in \R^+$ controls the deviation of the funnel around the circumvent function.

Further, the non-smooth $\sign$ function is approximated by the smooth function $\tanh$. When the conflict is resolved, $\theta_i(t)$ becomes 0 and the second part decays $\alpha_i(t)$ exponentially back to zero with a rate of decay governed by constant $\kappa$. An example of how the circumvent function modifies the funnel is shown in Figure \ref{fig:Funnel} (c).

% \begin{figure*}[h]
%     \centering
%     \includegraphics[scale=0.6]{Plots/Fun.eps}
%     \caption{Revamped Funnel}
%     \label{fig:Modif}
% \end{figure*}

Let us now define $\gamma_L = \col (\gamma_{1,L}, \ldots, \gamma_{n,L})$, $\gamma_U = \col (\gamma_{1,U}, \ldots, \gamma_{n,U})$, $\gamma_d = \textsf{diag} (\gamma_{1,U}-\gamma_{1,L}, \ldots, \gamma_{n,U}-\gamma_{n,L})$, and $\gamma_s = \col (\gamma_{1,U}+\gamma_{1,L}, \ldots, \gamma_{n,U}+\gamma_{n,L})$.
% $\gamma_{s} = \col (\gamma_{1,s}, \ldots, \gamma_{n,s})$, where, $\gamma_{i,s} = \gamma_{1,U}+\gamma_{1,L}$.
% $\gamma_d = \textsf{diag}(\gamma_{1,d}, \ldots, \gamma_{n,d})$, where, $\gamma_{i,d}=\gamma_{i,U}-\gamma_{i,L}, \forall i \in [1;n]$.

\begin{lemma} \label{lem:funnel}
     %As derived in Section \ref{Funnel} and according to the definitions of $\gamma_s(t)$ and $\gamma_d(t)$, we infer that 
     $\gamma_s(t), \dot{\gamma}_s(t), \gamma_d(t), \dot{\gamma}_d(t) \in \Ce$. 
\end{lemma}
% \begin{proof}
%     The detailed proof is available in *Archive Paper*.
% \end{proof}
\begin{proof}
    From definitions \eqref{eqn:perform},\eqref{eqn:bump1}, and \eqref{eqn:bump2}, one has  $\rho(t) \in \Ce$ and $\beta^{\hat{j}}(t) \in \Ce$. Thus, to show that $\gamma_s(t), \gamma_d(t), \dot{\gamma}_s(t)$ and $\dot{\gamma}_d(t) \in \Ce$, it is sufficient to show that $\alpha(t), \dot{\alpha}(t)$ $\in \Ce$.

    Since $\eta(t), \rho(t), \beta^{\hat{j}}(t) \in \Ce$ and $\mu > 0$ is a bounded tolerance, $\psi(t) = \col(\psi_1(t), \ldots, \psi_n(t))$ is also continuous and bounded. Further, $\dot{\psi}(t) = \col(\dot{\psi}_1(t), \ldots, \dot{\psi}_n(t))\in \Ce$. Hence, $\psi(t), \dot{\psi}(t) \in \Ce$.

    Now, depending on the sign of $\psi(t)$, consider the two cases and look at $\alpha(t)$ and $\dot{\alpha}(t)$ elementwise:
   
    \textbf{Case I.} [$\psi_i(t) \geq 0$] This implies that $\sign(\psi_i(t)) = 1 $ and $ \theta_i(t) = 0$. Thus, $\dot{\alpha}_i(t) = -\kappa \alpha_i(t)\in\Ce$ which implies $\alpha_i(t) = \alpha_i(0)e^{-\kappa t}\in \Ce$ 
    % \begin{align*} 
    %     & \dot{\alpha}_i(t) = -\kappa \alpha_i(t) \implies \dot{\alpha}_i(t) \in \Ce, \\
    %     & \alpha_i(t) = \alpha_i(0)e^{-\kappa t} \implies \alpha_i(t) \in \Ce.
    % \end{align*}
    for $i \in [1;n]$.
    
    \textbf{Case II.} [$\psi_i(t)<0$] This implies that $\sign(\psi_i(t)) = -1$ and $\theta_i(t) = 2\theta_o$. 
    \begin{align} \label{cb:alpha2}
        \dot{\alpha}_i(t) = \frac{2\theta_o}{\psi_i(t) + \alpha_i(t)} -\kappa \alpha_i(t).
    \end{align}

    We will prove the boundedness of $\dot{\alpha}_i(t)$ by contradiction.
    Let $\psi_i(t) + \alpha_i(t) \rightarrow 0$ (converges to zero). Then taking its derivative w.r.t time $t$, we can say $\dot{\psi}_i(t) + \dot{\alpha}_i(t) \in \Ce$. From the fact that $\dot{\psi}_i(t) \in \Ce$, we have  $\dot{\alpha}_i(t)\in\Ce$. 
    However, from (\ref{cb:alpha2}) we can observe that if $\psi_i(t) + \alpha_i(t) \rightarrow 0$ then $\dot{\alpha}_i(t) \rightarrow \infty$. This leads to a contradiction. Therefore, $\psi_i(t) + \alpha_i(t) \nrightarrow 0$ (does not converge to zero) and consequently, $\dot{\alpha}_i(t) \in \Ce$ for $i \in [1;n]$.
    
    To further prove the boundedness of $\alpha_i(t)$, we will again use contradiction.
    Let $\alpha_i(t) \rightarrow \infty$. Now, since $\psi_i(t)$ is bounded and $2\theta_o$ is a finite constant, $\frac{2\theta_o}{\psi_i(t)+\alpha_i(t)} \rightarrow 0 \implies \dot{\alpha}_i(t) = -\kappa \alpha_i(t) \rightarrow -\infty$. But $\alpha_i(t) \rightarrow \infty$ and $\dot{\alpha}_i(t) \rightarrow -\infty$ are contradictory. Hence, $
        \alpha_i(t) \nrightarrow \infty \text{ for } i \in [1;n].$
    
    Let $\alpha_i(t) \rightarrow -\infty$. Similarly, since $\psi_i(t)$ is bounded and $2\theta_o$ is a finite constant, $\frac{2\theta_o}{\psi_i(t)+\alpha_i(t)} \rightarrow 0 \implies \dot{\alpha}_i(t) = -\kappa \alpha_i(t) \rightarrow \infty$. But $\alpha_i(t) \rightarrow -\infty$ and $\dot{\alpha}_i(t) \rightarrow \infty$ are contradictory. Hence, $\alpha_i(t) \nrightarrow -\infty \text{ for } i \in [1;n]$.\\
    Therefore, in both cases we reach the same conclusion, $\alpha(t), \dot{\alpha}(t) \in \Ce$.
\end{proof}

\subsection{Controller Design}\label{sec:control}
In this section, utilizing the adaptive funnel, discussed in the previous section, we derive the funnel control law to solve Problem \ref{prob1}. The controller design is done in three stages.

\textbf{Stage I.} Given an initial state $x(0)$ and target state $\T$, construct the funnel constraints \eqref{eqn:ppc} that guide the system trajectory to the target $\T$, as discussed in Section \ref{sec:reach}.

\textbf{Stage II.} Given the unsafe region $\U$, obtain $\Obs^{\hat{j}}$ as shown in (\ref{eqn:firstobs}) and compute the circumvent function according to (\ref{eqn:bump1}) or (\ref{eqn:bump2}). Now modify the funnel around the circumvent function, as discussed in Section \ref{Funnel}, and determine the adaptive funnel framework, defined by $\gamma_{L}$ and $\gamma_{U}$.

\textbf{Stage III.} For the modified funnel, we define the normalized error as
\begin{align}
    \hat{e}(x,t) = 2\gamma_d(t)^{-1} \left(x-\frac{1}{2}\gamma_s(t)\right). \label{eqn:e_ras}
\end{align}

The corresponding constrained region $\hat{\mathbb{D}}$ can be represented by:
$\hat{\mathbb{D}} := \{\hat{e}(x,t) : \hat{e}(x,t) \in (-1, 1)^n\}$. 
The transformed error is defined as:
\begin{align}
    &\hat{\varepsilon}(x,t) = y(\hat{e}(x,t)) \nonumber \\
                    &=\hspace{-0.4em} \col \hspace{-0.2em}\left(\hspace{-0.2em} \ln\left(\frac{1\hspace{-0.2em}+\hspace{-0.2em}\hat{e}_1(x_1,t)}{1\hspace{-0.2em}-\hspace{-0.2em}\hat{e}_1(x_1,t)}\right),\ldots,\ln\left(\frac{1\hspace{-0.2em}+\hspace{-0.2em}\hat{e}_n(x_n,t)}{1\hspace{-0.2em}-\hspace{-0.2em}\hat{e}_n(x_n,t)}\right)\hspace{-0.2em} \right). \label{eqn:eps_ras}
\end{align}
We also define a diagonal matrix, $\hat{\xi}(x,t)$, as
\begin{align}
    \hat{\xi}(x,t) = \frac{4 \gamma_d^{-1}}{(1-\hat{e}^{T}(x,t)\hat{e}(x,t))}. \label{eqn:xi_ras}
\end{align}

Now, in Theorem \ref{theorem_ras}, we propose a control strategy $\hat{u}(x,t)$ such that the state trajectory satisfies reach-avoid-stay specifications.

\begin{theorem} \label{theorem_ras}
    Consider a nonlinear control-affine system $\mathcal{S}$ given in (\ref{eqn:sysdyn}), assigned a reach-avoid task expressed mathematically through (\ref{eqn:perform}) and (\ref{eqn:bump1}, \ref{eqn:bump2}) respectively. If the initial state $x(0)$ is within the modified funnel (Section \ref{Funnel}), then the control strategy
    \begin{multline} \label{eqn:Control_ras}
        \hat{u}(x,t) = -g(x)^T(g(x) g(x)^T)^{-1} \\
        \left(\hat{k}\hat{\xi}(x,t) \hat{\varepsilon}(x,t) - \frac{1}{2}\dot{\gamma}_d(t) \hat{e}(x,t) \right).
    \end{multline}
    will drive the state trajectory $x(t)$ to the target $\T$ while avoiding the unsafe set $\Obs^{\hat{j}}$ (\ref{eqn:firstobs}) and adhering to state constraints, i.e., $\exists t \in \R_0^+ : x(t) \in \T$ and $\forall t \in \R_0^+, x(t) \notin \Obs^{\hat{j}} \text{ and } x(t) \in \X$. Here, $\hat{k}$ is any positive constant, $\hat{e}(x,t)$, $\hat{\varepsilon}(x,t)$, and $\hat{\xi}(x,t)$ are defined in (\ref{eqn:e_ras}), (\ref{eqn:eps_ras}), and (\ref{eqn:xi_ras}), respectively.
\end{theorem}

\begin{proof}
% For simplicity of paper, we will denote the timed trajectories without explicitly mentioning the t, i.e., replace $x(t)$ with $x$.
The proof comprises three steps. 
First, we show that there exists a maximal solution for the normalized error $\hat{e}(x,t)$, which implies that $\hat{e}(x,t)$ remains within $\hat{\mathbb{D}}$ in the maximal time solution interval $[0, \tau_{\max})$. Next, show that the proposed control law (\ref{eqn:Control_ras}) constraints $\hat{e}(x,t)$ to a compact subset of $\hat{\mathbb{D}}$. Finally, prove that $\tau_{\max}$ can be extended to $\infty$.

Before proceeding let us introduce two lemmas:

\begin{lemma}\cite[Theorem 54]{sontag} \label{lem:54}
    Consider the IVP $\dot{y} = H(y,t), y(0) \in \mathbb{D}_y$.
    Assume $H:\mathbb{D}_y \times \mathbb{R}_{>0} \rightarrow \mathbb{R}$ is
    \begin{enumerate}
        \item locally Lipschitz on $y$, for each $t \in \mathbb{R}_{>0}$
        \item piecewise continuous on $t$ for each fixed $y \in \mathbb{D}_y$
    \end{enumerate}
    Then there exists a unique and maximal solution
    $y:[0, \tau_{\max}) \rightarrow \mathbb{D}_y,$
    where $\tau_{\max} \in \mathbb{R}_{>0} \cup \infty$.
\end{lemma}

\begin{lemma}\cite[Proposition C.3.6]{sontag} \label{lem:c.3.6}
    Consider all the assumptions of Lemma \ref{lem:54} to hold true.
    For a maximal solution $y$ on $[0, \tau_{\max})$ with $\tau_{\max} < \infty$ and for any compact set $\mathbb{D}'_y \in \mathbb{D}_y$, 
    $\exists t' \in [0, \tau_{\max}) \text{, such that } y(t') \notin \mathbb{D}_y.$
\end{lemma}

Continuing with the proof.

\textbf{Step 1.} Taking derivatives of (\ref{eqn:e_ras}) and (\ref{eqn:eps_ras}), we have:
\begin{align}
    \dot{\hat{e}} = 2\gamma_d^{-1} \left( \dot{x} - \frac{1}{2}\dot{\gamma}_s - \frac{1}{2}\dot{\gamma}_d \hat{e} \right), \label{eqn:edot} \text{ and }
    \dot{\hat{\varepsilon}} = \frac{2}{1-\hat{e}^T\hat{e}}\dot{\hat{e}}.
\end{align}

Substituting the controller (\ref{eqn:Control_ras}) in the system dynamics (\ref{eqn:sysdyn}), we obtain the closed-loop dynamics:
\begin{align*}
    \dot{x} &= H_1(x,\hat{e},t) := f(x) + \left(-\hat{k} \hat{\xi} \hat{\varepsilon} - \frac{1}{2} \dot{\gamma}_d \hat{e} \right) \nonumber 
            %&= f(x) + \left(-\hat{k} \frac{4 \gamma_d^{-1}}{1-\hat{e}^T\hat{e}} \text{ln} \left( \frac{1+\hat{e}}{1-\hat{e}} \right) - \frac{1}{2} \dot{\gamma}_d \hat{e} \right) \label{eqn:H1}
\end{align*}
and substituting the above equation in $\dot{\hat{e}}$, we obtain
\begin{align*}
    \dot{\hat{e}} = H_2(x,\hat{e},t) := 2\gamma_d^{-1} \left( H_1(x,\hat{e},t) - \frac{1}{2}\dot{\gamma}_s - \frac{1}{2}\dot{\gamma}_d \hat{e} \right).
\end{align*}
Consider the augmented state $y$ and its derivative $\dot{y}$ as
\begin{align*}
    y = \begin{bmatrix}
x\\
\hat{e}
\end{bmatrix} \text{, } 
\dot{y} = H(y,t) := \begin{bmatrix}
H_1(x,\hat{e},t)\\
H_2(x,\hat{e},t)
\end{bmatrix}.
\end{align*} 

Since the initial state $x(0)$ is within the updated funnel, the initial normalized error $\hat{e}(x(0), 0)$ is within the constrained region $\hat{\mathbb{D}}$. Note, that $\hat{\mathbb{D}}$ is an open and bounded set. Further, define $\hat{\mathbb{D}}_x := \{x \in \mathbb{R}^n | \hat{e}(x(0),0) \in \hat{\mathbb{D}}\}$, which is a non-empty open and bounded set. Thus, $\hat{\mathbb{D}}_y := \hat{\mathbb{D}}_x \times \hat{\mathbb{D}}$ is also a non-empty open and bounded set and the initial condition of the augmented state satisfy $y(0) = \begin{bmatrix}
x(0)\\
\hat{e}(x(0),0)
\end{bmatrix} \in \hat{\mathbb{D}}_y$.
Therefore, we have the following initial value problem at hand:
    $\dot{y} = H(y,t), y(0) \in \hat{\mathbb{D}}_y$.

We can see that $\hat{\varepsilon}$ (\ref{eqn:eps_ras}), $\hat{\xi}$ (\ref{eqn:xi_ras}) and $\dot{\gamma}_d \hat{e}$, defined on $\hat{\mathbb{D}}_y$, are locally Lipschitz continuous in $\hat{e}$. Further, according to Assumption \ref{assum:lip}, $f(x)$ is also Lipschitz continuous on $\hat{\mathbb{D}}_y$ in $x$. Therefore, we can conclude that $H(y,t)$ is locally Lipschitz continuous on $\hat{\mathbb{D}}_y$ in $y$.

Hence, according to Lemma \ref{lem:54}, there exists a maximal solution of the IVP $\dot{y} = H(y,t), y(0) \in \hat{\mathbb{D}}_y$ in the time interval $[0, \tau_{\max})$:
%\begin{align}
    $y(t) \in \hat{\mathbb{D}}_y \forall t \in [0, \tau_{\max}).$ %\nonumber
%\end{align}

%%%%%%%%%%%%%%%%%%%%%%%%%%%%%%%%%%%%%%%%%%%%%%%%%%%%%%%%%%%%%%%%%%%%%%

\textbf{Step 2.} Based on Step 1, we know
\begin{align*}
             &y(t) \in \hat{\mathbb{D}}_y, \forall t \in [0, \tau_{\max}) \\
    \implies &\hat{e}(t) \in \hat{\mathbb{D}}, \forall t \in [0, \tau_{\max}) \\
    \implies &\gamma_L(t) < x(t) < \gamma_U(t), \forall t \in [0, \tau_{\max}).
\end{align*}

Consider the following positive definite and radially unbounded Lyapunov function candidate: $V = \frac{1}{2}\hat{\varepsilon}^T\hat{\varepsilon}$. 

Differentiating $V$ with respect to time $t$ and substituting $\dot{\hat{\varepsilon}}$, $\dot{\hat{e}}$ and system dynamics \eqref{eqn:sysdyn}, we obtain:
\begin{align*}
    \dot{V} &= \hat{\varepsilon}^T \dot{\hat{\varepsilon}}= \hat{\varepsilon}^T \frac{2}{1-\hat{e}^T\hat{e}}\dot{\hat{e}} = \hat{\varepsilon}^T \hat{\xi} \left(\dot{x} - \frac{1}{2}(\dot{\gamma}_{s}-\dot{\gamma}_{d}\hat{e})\right) \\
            &= \hat{\varepsilon}^T \hat{\xi} \left( f(x) + g(x)u  - \frac{1}{2}(\dot{\gamma}_{s}-\dot{\gamma}_{d}\hat{e}) \right).
\end{align*}
Now employ the control strategy (\ref{eqn:Control_ras}), we get
\begin{align*}
    \dot{V} &= \hat{\varepsilon}^T \hat{\xi} \left( f(x) + \left(-k \hat{\xi} \hat{\varepsilon} - \frac{1}{2}\dot{\gamma}_d \hat{e}\right)  - \frac{1}{2}(\dot{\gamma}_{s}-\dot{\gamma}_{d} \hat{e}) \right) \\
            &= \hat{\varepsilon}^T \hat{\xi} \left( -k \hat{\xi} \hat{\varepsilon} + \left( f(x) - \frac{1}{2}\dot{\gamma}_{s} \right) \right)\\ &\leq \norm{\hat{\varepsilon}^T \hat{\xi} \left( -k \hat{\xi} \hat{\varepsilon} + \left( f(x) - \frac{1}{2}\dot{\gamma}_{s} \right) \right)} \\
            &\leq -k\|\hat{\varepsilon}\|^2 \|\hat{\xi}\|^2 + \|\hat{\varepsilon}\| \|\hat{\xi}\| \|\hat{\Phi}\|,
\end{align*}
where $\hat{\Phi} := f(x) - \frac{1}{2}\dot{\gamma}_s$. We will look at the boundedness of the two terms in $\hat{\Phi}$ separately. First, we know $f(x)$ is a continuous function of $x$ and $x \in \hat{\mathbb{D}}_x, \forall t \in [0, \tau_{\max})$, an open and bounded set. Thus, by applying the extreme value theorem, we can infer $\|f(x)\| < \infty$. 
%Secondly, by definition, $w$ is bounded. 
Finally, from Lemma \ref{lem:funnel} we know that $\dot{\gamma}_s$ is also bounded. Hence, $\hat{\Phi} \in \Ce, \forall t \in [0, \tau_{max}]$.

Now add and substract $k\hat{\theta}\norm{\hat{\varepsilon}}^2 \|\hat{\xi}\|^2$, where $0<\hat{\theta}<1$
\begin{align*}
    \dot{V} &\leq -k(1-\hat{\theta})\norm{\hat{\varepsilon}}^2 \|\hat{\xi}\|^2 - \norm{\hat{\varepsilon}}\|\hat{\xi}\| \left(k \hat{\theta} \norm{\hat{\varepsilon}}\|\hat{\xi}\| - \|\hat{\Phi}\| \right) \\
            &\leq -k(1-\hat{\theta})\norm{\hat{\varepsilon}}^2 \norm{\hat{\xi}}^2, \forall k \hat{\theta} \norm{\hat{\varepsilon}}\|\hat{\xi}\| - \|\hat{\Phi}\| \geq 0 \\
            &\leq -k(1-\hat{\theta})\norm{\hat{\varepsilon}}^2 \|\hat{\xi}\|^2, \forall \norm{\hat{\varepsilon}} \geq \frac{\|\hat{\Phi}\|}{k \hat{\theta} \|\hat{\xi}\|} \nonumber, \forall t \in [0, \tau_{\max}).
\end{align*}
Therefore, we can conclude that there exists a time-independent upper bound $\hat{\varepsilon}^* \in \mathbb{R}_{0}^+$ to the transformed error $\hat{\varepsilon}$, i.e., $\|\hat{\varepsilon}\| \leq \hat{\varepsilon}^* \forall t \in [0, \tau_{\max})$.

Further, we know from (\ref{eqn:eps_ras}) that $\hat{\varepsilon}_i = \ln \left( \frac{1+\hat{e}_i}{1-\hat{e}_i} \right)$. Taking inverse, we can bound the normalized error $\hat{e}(x,t) = \col(\hat{e}_1, \ldots, \hat{e}_n)$ as:
\begin{align*}
    -1 < \frac{\hat{e}_i^{-\hat{\varepsilon}_i^*}-1}{\hat{e}_i^{-\hat{\varepsilon}_i^*}+1} =: \hat{e}_{i,L} \leq \hat{e}_i \leq \hat{e}_{i,U} := \frac{\hat{e}_i^{\hat{\varepsilon}_i^*}-1}{\hat{e}_i^{\hat{\varepsilon}_i^*}+1} < 1 \nonumber \\
    \forall t \in [0, \tau_{\max}), \ \text{for } i \in [1;n].
\end{align*}
Therefore, by employing the control law (\ref{eqn:Control_ras}), we can constrain $\hat{e}$ to a compact subset of $\hat{\mathbb{D}}$ as:
\begin{align} \label{eqn:e_compact}
    \hat{e}(x,t) \in [\hat{e}_L, \hat{e}_U] =: \hat{\mathbb{D}}' \subset \hat{\mathbb{D}}, \forall t \in [0, \tau_{\max}),
\end{align}
where, $\hat{e}_L = \col(\hat{e}_{1,L}, \ldots, \hat{e}_{n,L})$ and $\hat{e}_U = \col(\hat{e}_{1,U}, \ldots, \hat{e}_{n,U})$
%%%%%%%%%%%%%%%%%%%%%%%%%%%%%%%%%%%%%%%%%%%%%%%%%%%%%%%%%%%%%%%%%%%%%%

\textbf{Step 3.} Finally, we prove that $\tau_{\max}$ can be extended to $\infty$. 

We know that $\hat{e}(x,t) \in \hat{\mathbb{D}}', \forall t \in [0, \tau_{\max})$, where $\hat{\mathbb{D}}'$ is a non-empty compact subset of $\hat{\mathbb{D}}$.

Consequently we can conclude that $x(t) = \frac{1}{2}\gamma_d \hat{e} + \gamma_s$ also evolves in a compact set:
\begin{align}
    x(t) \in \hat{\mathbb{D}}'_x \subset \hat{\mathbb{D}}_x, \forall t \in [0, \tau_{\max}).
\end{align}
Define the compact set $\hat{\mathbb{D}}'_y := \hat{\mathbb{D}}'_x \times \hat{\mathbb{D}}'$ and note that $\hat{\mathbb{D}}'_y \subset \hat{\mathbb{D}}$. Therefore, there is no $t \in [0, \tau_{\max})$ such that $y(t) \notin \hat{\mathbb{D}}_y$.

However, if $\tau_{\max} < \infty$ then according to Lemma \ref{lem:c.3.6}, $\exists t' \in [0, \tau_{\max})$ such that $y(t) \notin \hat{\mathbb{D}}_y$. This leads to a contradiction!
Hence, we conclude that $\tau_{\max}$ can be extended to $\infty$, i.e., $x(t)$ satisfies the funnel constraints in (\ref{eqn:ppc}) $\forall t \geq 0$.

In conclusion, the satisfaction of \eqref{eqn:ppc} is guaranteed for all time when we employ the control strategy (\ref{eqn:Control_ras}).
\end{proof}

\begin{remark}
    From Assumption \ref{assum:lip}, we know $g(x)g^T(x)$ is invertible. (\ref{eqn:e_compact}) entails that $\hat{e}$ is bounded. And by definitions \eqref{eqn:e_ras} and \eqref{eqn:xi_ras}, $\hat{\xi}$ and $\hat{\varepsilon}$ are also bounded. Further, from Lemma \ref{lem:funnel}, $\dot{\gamma}_d$ also $\in \Ce$. Finally, all the non-smooth functions in the revamped funnel design in Section \ref{Funnel} are replaced by their smooth approximations.
    Hence, the control law $\hat{u}(x,t)$ (\ref{eqn:Control_ras}) is well-defined, i.e., continuous, smooth, and bounded.
\end{remark}

\begin{remark}
    The structure of the controller defined in Theorem \ref{thm:control} \eqref{eqn:Control_rs} is the same as that in \eqref{eqn:Control_ras}, only with the modified funnel constraints.
\end{remark}

In Figure \ref{fig:3d}, we present a 3D visualization of a scenario where $\X \subset \R^2$ and the modified funnel circumvents around the unsafe set $\U$, providing a safe path for the trajectory to reach the target $\T \subset \X$ while staying clear of $\U$.

\begin{figure}[h]
    \centering
    \includegraphics[width=0.5\textwidth]{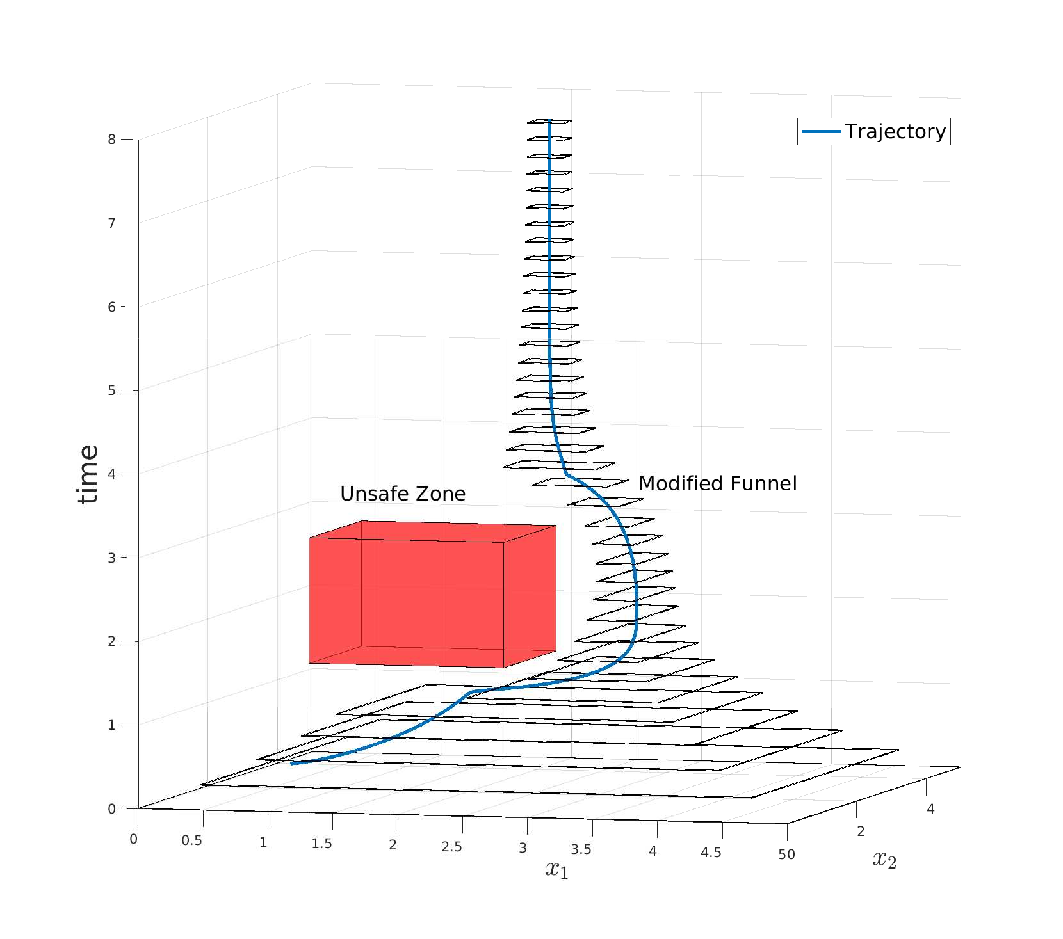}
    \vspace{-0.8cm}
    \caption{3D visualization.}
    \label{fig:3d}
\end{figure}

\begin{figure*}[t]
    \centering
    \includegraphics[width=\textwidth]{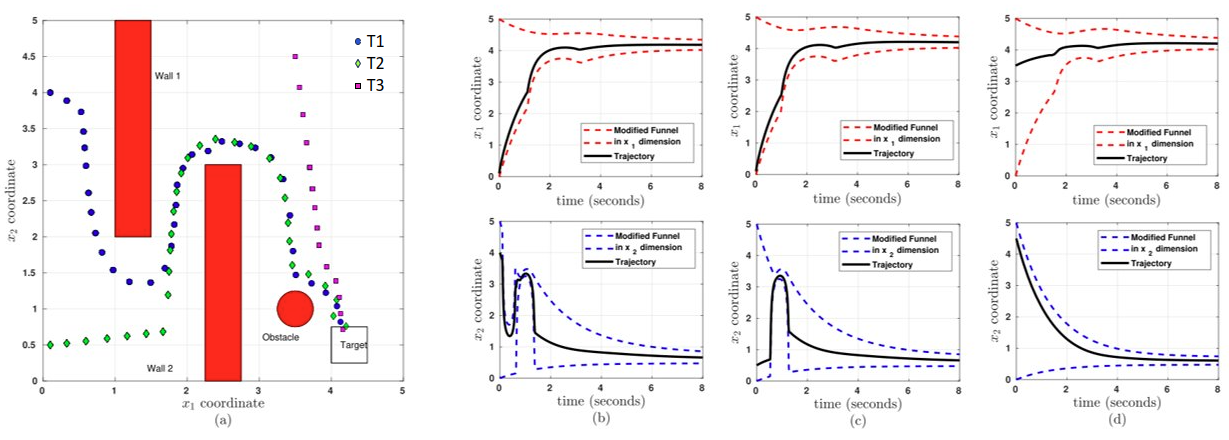}
    \vspace{-0.8cm}
    \caption{(a) Starting from three different initial states, we have three different trajectories: T1 (blue circle), T2 (green diamond), and T3 (magenta square). Adaptive funnel framework with controlled system trajectories for (b) T1, (c) T2, and (d) T3.}
    \label{fig:sim}
\end{figure*}

\section{Extension to Tackle General Unsafe sets} \label{sec:multiobs}
Given the unsafe region $\U$ with $n_u$ connected convex sets, we first choose $\Obs^{\hat{j}}$ (\ref{eqn:firstobs}). The control law $\hat{u}(x,t)$ (\ref{eqn:Control_ras}) ensures that the controlled state trajectory reaches the target while avoiding this $\Obs^{\hat{j}}$. Now, to address $\Obs^j$ for $j={1,2,\ldots,n_u}$, we iterate through this procedure until the controlled system trajectory stays entirely clear of the unsafe region $\U$. 

Further, in each iteration, for defining the $\beta$ function, we have a certain degree of randomness. We randomly select $i^{\hat{j}}$ from all the possible alternatives (\ref{eqn:firstdim}). Moreover, we also randomly choose whether to introduce $\beta_i^{\hat{j}}$ in the lower constraint (\ref{eqn:bump1}) or upper constraint (\ref{eqn:bump2}), as discussed in Section \ref{sec:bump}. This randomness allows exploration of all the paths around unsafe regions, thus resulting in a higher probability of obtaining a closed-form controller satisfying reach-avoid-stay specifications for complex environments. The algorithm is presented in Algorithm 1. 

\vspace{0.6cm}
\hrule
\vspace{0.1cm}
\hspace{-0.45cm} \textbf{Algorithm 1} Extension for general unsafe region
\vspace{0.1cm}
\hrule
\vspace{0.2cm}
\hspace{-0.45cm} \textbf{Input:} $\X, \T, \U=\{\Obs_1, \Obs_2, \ldots, \Obs_{n_u}\}, x(0)$ \\
\textbf{Output:} $\hat{u}(x(0),\X,\T,\U,x,t): \{ \exists \tau \in \R^+_0: x(\tau) \in \T \text{ and } \forall t \in \R^+_0: x(t) \in \X, x(t) \cap \U = \emptyset \}$
\begin{enumerate}
    \item[1.] Given $\T$, choose $\eta \in int(\T)$ and construct funnel constraints to enforce reachability (\ref{eqn:ppc})
    \item[2.] Apply control law $u(x,t)$ (\ref{eqn:Control_rs}) to drive the controlled trajectory $x(t)$ to the target while remaining within the state limits.
    \item[3.] \textbf{while} true \textbf{do}
    \item[4.] \begin{itemize}
      \item[] Obtain obstacle $\Obs_{\hat{j}} \in \U$ (\ref{eqn:firstobs}) and introduce the circumvent function $\beta$ (\ref{eqn:bump1}) or (\ref{eqn:bump2}) to modify the funnel around the obstacle as discussed in Section \ref{Funnel}.
    \end{itemize}
    \item[5.] \begin{itemize}
      \item[] Apply control law $\hat{u}(x,t)$ \eqref{eqn:Control_ras} and obtain the controlled trajectory $x_u(t)$.
    \end{itemize}
    \item[6.] \begin{itemize}
      \item[] \textbf{if} ($x_u(t) \cap \Obs_{j} = \emptyset, \forall j \in [1;n_u]$)
    \end{itemize}
    \item[7.] \begin{itemize}
      \item[] \hspace{0.4cm} \textbf{return} $\hat{u}(x,t)$
    \end{itemize}
    \item[8.] \begin{itemize}
      \item[] \textbf{end}
    \end{itemize}
    \item[9.] \textbf{end}
\end{enumerate}
\vspace{0.1cm}
\hrule

\vspace{0.5cm}

\begin{corollary}
    Thus, given a system $\mathcal{S}$ in (\ref{eqn:sysdyn}), target set $\T$ in the state space $\X$ and unsafe region $\U$, termination of the Algorithm 1 defines an adaptive funnel framework and provides us a well-defined closed-form control law (\ref{eqn:Control_ras}) that will guide the system trajectory to the target while avoiding the unsafe region, enforcing reach-avoid-stay specifications.
\end{corollary}

% \begin{remark}
%     Although termination of the loop is not guaranteed, it has been proven that if the loop terminates, then the reach-avoid-stay specification will be necessarily satisfied.
% \end{remark}

A simulation study illustrating the efficacy of the algorithm in solving reach-avoid-stay specifications in a multi-obstacle environment is presented in the next section.

%%%%%%%%%%%%%%%%%%%%%%%%%%%%%%%%%%%%%%%%%%%%%%%%%%%%%%%%%%%%%%%%%%%%%%%%%%%%%%%%
%%--------------------------------NEW SECTION---------------------------------%%
%%%%%%%%%%%%%%%%%%%%%%%%%%%%%%%%%%%%%%%%%%%%%%%%%%%%%%%%%%%%%%%%%%%%%%%%%%%%%%%%
\section{Simulation Results}\label{sec:sim}

Consider a three-wheeled omnidirectional robot operating on a 2-D plane. The Kinematic model of the mobile robot is expressed as:
\begin{align}
    \begin{bmatrix}
        \dot{x} \\ \dot{y} \\ \dot{\theta}
    \end{bmatrix}
    = 
    \begin{bmatrix}
        \cos{\theta} & \sin{\theta} & 0 \\ \sin{\theta} & -\cos{\theta} & 0 \\ 0 & 0 & 1
    \end{bmatrix}
    \begin{bmatrix}
        u \\ v \\ \omega
    \end{bmatrix},
\end{align}
where $(x, y) \text{ and } \theta$ captures the robot's position and orientation respectively. The control inputs, u, v, and $\omega$ are linear velocities in the x and y direction of the robot frame and the angular velocity respectively. Note that the robot dynamics satisfy Assumption \ref{assum:lip}.

We ran the tests for a 2D arena with two wall obstacles and a circular obstacle. The funnel for guiding the robot towards the target is shaped according to \eqref{eqn:perform} with the following parameters: $\rho_{i,0} = 1, \rho_{i,\infty} = 0.05$ and $l_i = 0.7$ for $i \in {1,2}$. $i=1$ and $i=2$ represents the $x_1$ and $x_2$ coordinates respectively. For the circumvent function \eqref{eqn:bump1} or \eqref{eqn:bump2}, we define $k=0.001$, $\delta B = 0$ and $\delta t = 0.1$ for all the obstacles. Finally the adaptive law \eqref{eqn:adap} is established with $\mu = 10$, $\kappa = 0.3$ and $\theta_o = 0.1$.

% \textbf{Task 1.} A three-wheeled omnidirectional robot is initially placed at $x = [1; 2]$ units and is designated the task of reaching the target at $T = [4.25; 3.75]$ units. An obstacle of dimension $1 \text{ unit} \times 1 \text{ unit}$ is placed at $[2.5; 3]$ units. 

% \textbf{Task 2.} In the second task, the same robot is initially placed at $x = [1; 4]$ units and is designated the task of reaching the target at $T = [4.25; 4]$ units. Meanwhile there is a wall of dimension $0.5 \text{ units} \times 2 \text{ units}$ in between at $[2.25, 4]$ units.

% \textbf{Task 3.} Finally, in the third task the same robot is now assigned to dodge multiple obstacles in the arena. The initial position and the target are same as that in Task 1. Two obstacles of dimensions $0.5 \text{ units} \times 0.5 \text{ units}$ are introduced in the arena at $[2.25, 4]$ units and $[2.25, 4]$ units, respectively.

The simulation results with three different initial states are depicted in Figure \ref{fig:sim}. 

% \begin{figure*}[t]
%     \centering
%     \includegraphics[width=\textwidth]{Plots/RA.eps}
%     \caption{Reach target while avoiding an obstacle.}
%     \label{fig:simRA}
% \end{figure*}

% \begin{figure*}[t]
%     \centering
%     \includegraphics[width=\textwidth]{Plots/Wall.eps}
%     \caption{Reach target while avoiding a wall.}
%     \label{fig:simwall}
% \end{figure*}

% \begin{figure*}[t]
%     \centering
%     \includegraphics[width=\textwidth]{Plots/Multi.eps}
%     \caption{Reach target while avoiding multiple obstacles.}
%     \label{fig:simmulti}
% \end{figure*}

%%%%%%%%%%%%%%%%%%%%%%%%%%%%%%%%%%%%%%%%%%%%%%%%%%%%%%%%%%%%%%%%%%%%%%%%%%%%%%%%
%%--------------------------------NEW SECTION---------------------------------%%
%%%%%%%%%%%%%%%%%%%%%%%%%%%%%%%%%%%%%%%%%%%%%%%%%%%%%%%%%%%%%%%%%%%%%%%%%%%%%%%%
\section{Conclusion}\label{sec:conclusion}
In this work, we consider the controller synthesis problem for reach-avoid-stay specification. Given state space constraints, obstacles, and targets, we first proposed the introduction of a circumvent function and construction of an adaptive funnel framework. We have then derived a closed-form control law ensuring that the trajectories of a nonlinear system reach target while avoiding all the unsafe regions and respecting state-space constraints, thus, enforcing reach-avoid-stay specifications. Finally, the efficacy of the proposed approach is demonstrated through simulation studies. 
%Thus, the proposed framework provides a promising solution to address the challenges of achieving robust obstacle avoidance while satisfying prescribed performance requirements in various robotic applications.

%%%%%%%%%%%%%%%%%%%%%%%%%%%%%%%%%%%%%%%%%%%%%%%%%%%%%%%%%%%%%%%%%%%%%%%%%%%%%%%%

\addtolength{\textheight}{-12cm}   % This command serves to balance the column lengths

%%%%%%%%%%%%%%%%%%%%%%%%%%%%%%%%%%%%%%%%%%%%%%%%%%%%%%%%%%%%%%%%%%%%%%%%%%%%%%%%

\bibliographystyle{ieeetr} % We choose the "plain" reference style
\bibliography{sources} % Entries are in the refs.bib file

\end{document}